\documentclass[conference]{IEEEtran}
\IEEEoverridecommandlockouts
\usepackage{cite}
\usepackage{amsmath,amssymb,amsfonts,amsthm}

\newtheorem{theorem}{Theorem}
\newtheorem{remark}{Remark}%

\usepackage{algorithmic}
\usepackage{graphicx}
\usepackage{textcomp}
\usepackage{xcolor}
\def\BibTeX{{\rm B\kern-.05em{\sc i\kern-.025em b}\kern-.08em
    T\kern-.1667em\lower.7ex\hbox{E}\kern-.125emX}}
\begin{document}

\title{Opinion Dynamics Optimization Through Noncooperative Differential Games\\

\thanks{This work was supported by SGS, V\v{S}B - Technical University of Ostrava, Czech Republic, under grant No. SP2023/12 “Parallel processing of Big Data X”.}
}

\author{\IEEEauthorblockN{Hossein B. Jond}
\IEEEauthorblockA{\textit{Department of Computer Science} \\
\textit{V\v{S}B-Technical University of Ostrava}\\
Ostrava-Poruba, Czech Republic \\
hossein.barghi.jond@vsb.cz}

}

\maketitle

\begin{abstract}
In this paper, I study optimizing the opinion formation of a social network of a population of individuals on a graph whose opinion evolves according to the Hegselmann-Krause model for opinion dynamics. I propose an optimization problem based on a differential game for a population of individuals who are not stubborn. The objective of each individual is to seek an optimal control policy for her own opinion evolution by optimizing a personal performance index. The Nash equilibrium actions and the associated opinion trajectory with the equilibrium actions are derived for the opinion optimization model using Pontryagin's principle. The game strategies were executed on the well-known Zachary’s Karate Club social network. The resulting opinion trajectories associated with the game strategies showed that in non-stubborn Zachary’s network, the opinions moved toward the average opinion of the network, but a consensus of final opinions did not necessarily emerge.
\end{abstract}


\section{Introduction}
\label{Int}
Opinion dynamics, with its origins in sociology, is the study of the dynamical processes of public opinion formation, diffusion, and evolution~\cite{Liang,Friedkin}. Opinion dynamics have rapidly spread beyond sociology into other disciplines such as physics, mathematics, computer science, and control theory, to name a few. Among the various proposed models, the bounded confidence models of opinion dynamics have attracted more attention~\cite{Dittmer,Hegselmann-Krause}. The bounded confidence concept was incorporated into opinion dynamics by Hegselmann and Krause~\cite{Hegselmann-Krause}. In the Hegselmann-Krause (HK) model, a finite number of agents fuse their opinions only with those of others whose opinions do not differ more than their confidence bound. Several works have studied the extensions of the HK model under various settings~\cite{Etesami,Chen2020,Lorenz}. 

Opinion dynamics optimization is an emerging topic of interest for those dealing with social networks. In~\cite{Ghezelbash}, an optimization procedure is presented to select informed agents to prevent the influence of the annoying agents in the network. In another study, the problem of optimizing the placement of stubborn agents in a social network with the aim of maximally influencing the population was studied~\cite{Hunter}.

One interesting approach to studying opinion dynamics in social networks is a game-theoretic approach. In~\cite{Etesami}, the authors showed the HK model can be formulated as a sequence of the best response dynamics of a potential game. The work~\cite{Niazi} investigated a non-cooperative differential game model of opinion dynamics with an open-loop information structure. The opinion dynamics of a multiple-population social network were investigated through the application of a multiple-population mean field game in~\cite{Banez}. The work~\cite{Szolnoki} investigated the spread of opinions in a binary opinion model as an evolutionary game. A differential game in~\cite{YILDIZ} was used to model the opinion behavior of stubborn agents in a social network in the presence of a troll.

In this study, I model and evaluate the opinion behavior of selfish and self-interested agents in a social network whose opinion fusion rules conform to the HK opinion dynamics model and are constrained by a communication graph using non-cooperative differential games. The simulation results indicate that a network whose population has adopted game equilibrium strategies does not have a consensus of final opinions. 

The main contributions of this paper in comparison with similar works in \cite{Niazi,YILDIZ} are twofold. First, \cite{Niazi,YILDIZ} considered the single integrator dynamics model for opinion evolution, whereas I utilize the HK opinion dynamics model. The HK model, in contrast to the model with a single integrator, takes into account the topological structure of the social network. Second, the cost functions in \cite{Niazi,YILDIZ} involve minimizing disagreements in the network throughout the entire opinion formation process. In my optimization, I take final opinion disagreements into account. This choice of cost function results in more concise explicit expressions for the Nash equilibrium and the opinion trajectories associated with it. In addition, the distributed information structure of Nash equilibrium is observed. 

The paper is organized as follows. In Section~\ref{formulation}, I present a differential game model of the HK opinion dynamics for opinion optimization. In Section~\ref{main}, I derive the open-loop Nash equilibrium solution and the associated opinion trajectory with it for the optimization model. In addition, I also present an optimal control model of the HK opinion dynamics. In Section~\ref{simulation}, the results from the previous section are implemented on a real-world social network to observe the evolution of individuals' opinions. Conclusions and future works are discussed in Section~\ref{conclusions}.

\section{Opinion Dynamics Optimization}\label{formulation}

Consider a social network of $n$ agents indexed $1$ through $n$ on a communication graph $\mathcal{G}(\mathcal{V},\mathcal{E})$. The set of vertices $\mathcal{V}=\{1,\cdots,n\}$ corresponds to the set of agents. Each edge $(i,j) \in \mathcal{E}$ represents a mutual opinion flow between node $i$ and node $j$. The set of neighbors of vertex $i$ is denoted by $\mathcal{N}_i=\{j\in\mathcal{V}:(i,j) \mbox{ or }(j,i)\in \mathcal{E}, j\neq i\}$. I make the following assumption. The social graph $\mathcal{G}$ is connected. The connectivity of $\mathcal{G}$ means at least one globally reachable node (a root node of a spanning tree on the graph). The connected social graph means that each agent has at least one neighbor with whom they mutually interact (i.e., $\mathcal{N}_i\neq \emptyset$ or $|\mathcal{N}_i|\neq 0$ for all $i\in\mathcal{V}$).

Let $x_i(t)\in \mathit{\mathbb{R}}$ be the opinion of agent $i$ at time $t\in[0,t_f]$ where $t_f$ is a terminal time. In the HK model, the evolution of $x_i(t)$ at each stage $k=0,1,\cdots$ is as follows
\begin{equation}\label{eq:HK}
    x_i(k+1)=\frac{1}{|\mathcal{N}_i|}\sum_{j\in\mathcal{N}_i}x_j(k).
\end{equation}
In this model, each agent's opinion at each stage is the average opinion of her graph neighbors. The graph neighbors here are equivalent to the bounded confidence concept in the original HK model. An agent with the confidence bound $\epsilon_i$ takes only those agents $j$ into account whose opinions differ from her own not bigger than her confidence bound. The set of such agents $j$ is denoted by $N_i(t)=\{j\in\mathcal{V}:|x_i(t)-x_j(t)|\leq \epsilon_i\}$ which in terms of a time-invariant social graph $\mathcal{G}(\mathcal{V},\mathcal{E})$, is equivalent to the set of the graph neighbors for agent $i$ as defined before, i.e., $\mathcal{N}_i$. 
Another interpretation for $\mathcal{N}_i$ could be the set of agents whom a particular agent $i$ trusts to share and fuse opinion~\cite{BABAKHANBAK}.

Define $u_i(t)\in \mathit{\mathbb{R}}$ as agent $i$'s influence effort or simply her control input at time $t\in[0,t_f]$. The HK model (\ref{eq:HK}) with the
control input in continuous time is given by
\begin{align*}
    \dot{x}_i(t)&=\frac{1}{|\mathcal{N}_i|}\sum_{j\in\mathcal{N}_i}\big(x_j(t)-x_i(t)\big)+b_iu_i(t)\\&=\frac{1}{|\mathcal{N}_i|}\big(\sum_{j\in\mathcal{N}_i}x_j(t)-|\mathcal{N}_i|x_i(t)\big)+b_iu_i(t),
\end{align*}
or simply,
\begin{equation}\label{eq:dynamics0}
    \dot{x}_i(t)=\frac{1}{|\mathcal{N}_i|}\sum_{j\in\mathcal{N}_i}x_j(t)-x_i(t)+b_iu_i(t),
\end{equation}
where $b_i$ is a nonzero constant. 

The HK model assumes that the social network population is not stubborn, which means that the individuals do not hold any prejudices. In other words, their opinion evolves without any influence from their initial opinions, which suits modeling a social network population whose attitudes change independently of their previous beliefs. The process of opinion formation in a rational, non-stubborn social network at the individual level can be subject to each agent attempting to minimize her disagreement with her graph neighbors while, in the meantime, expending the least amount of influence effort. An appropriate cost function that characterizes such behavior or preferences in a social network for agent $i\in \mathcal{V}$ to minimize is  
\begin{equation} \label{eq:quadratic-cost0-nonstub}
J_i= \frac{1}{|\mathcal{N}_i|}\sum_{j\in\mathcal{N}_i}\big(x_i(t_f)-x_j(t_f)\big)^2+\int_{0}^{t_f} r_iu_i^2(t)~\mathrm{dt}, 
\end{equation}
where $r_i\in \mathit{\mathbb{R}}$ is a positive scalar ($r_i>0$). The cost function in (\ref{eq:quadratic-cost0-nonstub}) has two terms. The first term averages the sum of disagreements between the final opinions of each agent and her graph neighbors. The second term is the weighted control or influence effort made during the entire opinion formation process. 

In the control community, the optimization problem that emerged in this work is known as differential game problems \cite{Engwerda}. 
In the context of a differential game, each agent of the network is referred to as a player. In this context, each player seeks the control $u_i(t)$ that minimizes her cost function $J_i$ with the given initial opinions, subject to the continuous-time HK opinion evolution equation (\ref{eq:dynamics0}). In other words, the players in the game seek to minimize their cost functions in order to find their control or influence strategies $u_i(t)$ while their opinions evolve according to the differential equation (\ref{eq:dynamics0}). The behavior of self-interested players in the game of opinion formation is best reflected via noncooperative game theory. Under the framework of noncooperative games, the players can not make binding agreements, and therefore, the solution (i.e., the Nash equilibrium) has to be self-enforcing, meaning that once it is agreed upon, nobody has the incentive to deviate from \cite{Damme}. In the next section, I derive the open-loop Nash equilibrium solution for the optimization problem in (\ref{eq:dynamics0}) and (\ref{eq:quadratic-cost0-nonstub}).

\section{Main Result}\label{main}

Nash equilibrium is the main solution concept in noncooperative game scenarios. A Nash equilibrium is a strategy combination of all players in the game with the property that no one can gain a lower cost by unilaterally deviating from it. The open-loop Nash equilibrium is defined as a set of admissible actions ($u_1^*,\cdots,u_n^* $) if for all admissible ($u_1,\cdots,u_n$) the inequalities
$J_i (u_1^*,\cdots,u_{i-1}^*,u_i^*,u_{i+1}^*,\cdots,u_n^* )\leq
J_i (u_1^*,\cdots,u_{i-1}^*,u_i,u_{i+1}^*,\cdots,u_n^* )$
hold for $i\in\{1,\cdots,n\}$ where $u_i\in\Gamma_i$ and $\Gamma_i$ is the admissible strategy set for player $i$. The noncooperative differential game and the unique Nash equilibrium associated with it are discussed in \cite{Engwerda}.

In the following, I present the main result, which is the open-loop Nash equilibrium solutions and the associated opinion trajectories with the equilibrium actions for the previously introduced optimization problem. Before that, I define the following vectors and matrices to restate the optimization problem in a compact form. 

Define $A_i=[a_{ij}]$ where
\begin{equation*}
       a_{ij} =
\left\{
	\begin{array}{ll}
		1  & \mbox{if } (i,j)\in\mathcal{E} \mbox{ and } i\neq j,\\
		0  & \mbox{if } (i,j)\notin\mathcal{E} \mbox{ and } i\neq j, \\
            0 & \mbox{if } i=j, 
	\end{array}
\right.
\end{equation*}
to be the adjacency matrix for each agent $i\in\mathcal{V}$. The degree matrix is $D_i=\mathrm{diag}(0,\cdots,|\mathcal{N}_i|,\cdots,0)$ where "$\mathrm{diag}{·}$" stands for diagonal matrix. The graph Laplacian matrix for each agent $i$ is defined as
\begin{equation}\label{eq:Laplacian}
    L_i=D_i-A_i.
\end{equation}
The global adjacency, degree, and Laplacian matrices are $A=\sum_i A_i$,  $D=\sum_i D_i$, and $L=\sum_i L_i$. All the aforementioned matrices are symmetric. Additionally, let $W_i=\mathrm{diag}(0,\cdots,\omega_{ii},\cdots,0)$ and $W=\sum_i W_i$. Define vectors $x(t)=[x_1(t),\cdots,x_n(t)]^\top$, $B_i=[0,\cdots,b_i,\cdots,0]^\top$, and matrix $\Lambda=D^{-1}A-I$.

The optimization in (\ref{eq:dynamics0}) and (\ref{eq:quadratic-cost0-nonstub}) is restated in compact form as follows
\begin{align}\label{eq:nonstub-opt}
    & \min_{u_i}~  J_i(u_i)= \frac{1}{|\mathcal{N}_i|}x^\top(t_f)L_ix(t_f)+\int_{0}^{t_f} r_iu_i^2(t)~\mathrm{dt},\\
    &{s.t.}\notag\\
    & \dot{x}(t)=\Lambda x(t)+\sum_{i=1}^n B_iu_i(t), \quad x(0)=x_0, \quad i=1,\ldots,n. \nonumber
\end{align}

The presence of the Laplacian $L_i$ in the cost function above is due to its sum-of-squares property (see~\cite{Jond}).
\begin{equation*}
    x^\top(t_f) L_ix(t_f)=\sum_{j\in\mathcal{N}_i}\big(x_i(t_f)-x_j(t_f)\big)^2.
\end{equation*}

\begin{theorem}\label{theorem:nonstub}
 The unique Nash equilibrium actions and the associated opinion trajectory with these equilibrium actions for opinion formation of a social network as the noncooperative differential game in (\ref{eq:nonstub-opt}) are given by 
 \begin{align}
         u_i(t&)=-\frac{1}{r_i|\mathcal{N}_i|}B_i^\top\mathrm{e}^{(t_f-t)\Lambda^\top}L_iH^{-1}(t_f)\mathrm{e}^{t_f\Lambda}x_0, \label{eq:Nash-nonstub}\\
         x(t)&=\Big(\mathrm{e}^{t\Lambda}-\Psi(t)\Delta H^{-1}(t_f)\mathrm{e}^{t_f\Lambda}\Big)x_0,  \label{eq:Nash-tr-nonstub}
 \end{align}
 where
 \begin{align}
     &H(t_f)=I+\Psi(t_f)\Delta, \label{eq:matrix-H}\\ 
     &\Psi(t)=[\Psi_1(t),\cdots,\Psi_n(t)],\label{eq:matrix-Psi}\\
     &\Delta=[ \frac{1}{|\mathcal{N}_1|}L_1,\cdots, \frac{1}{|\mathcal{N}_n|}L_n]^\top,\label{eq:matrix-Delta} \\
     & \Psi_i(t)=\int_0^{t}\mathrm{e}^{(t-\tau)\Lambda}S_i\mathrm{e}^{(t-\tau)\Lambda^\top}~\mathrm{d\tau}, \quad S_i= \frac{1}{r_i}B_iB_i^\top.\label{eq:matrix-Psi-i}
 \end{align}
\end{theorem}

\begin{proof}
Define the Hamiltonian 
\begin{equation}\label{eq:Hamiltonian}
        \mathcal{H}_i=r_iu_i^2(t)+\lambda_i^\top(t)\Big(\Lambda x(t)+\sum_{i=1}^n B_iu_i(t)\Big), 
\end{equation}
where $\lambda_i(t)$ is a co-state vector.
According to Pontryagin’s principle, the necessary conditions for optimality are $\frac{\partial \mathcal{H}_i}{\partial u_i}=0$ and $\dot{\lambda}_i(t)=-\frac{\partial \mathcal{H}_i}{\partial x}$. Applying the necessary conditions on the Hamiltonian yield
\begin{align}
    &u_i(t)=-\frac{1}{r_i}B_i^\top\lambda_i(t), \label{eq:neccesary-u}\\
    &\dot{\lambda}_i(t)=-\Lambda^\top\lambda_i(t), \label{eq:neccesary-lam}
\end{align}
with the terminal condition 
\begin{equation}\label{eq:Lambda-compact}
    \lambda_i(t_f)=\frac{1}{|\mathcal{N}_i|}L_ix(t_f).
\end{equation}

The solution of (\ref{eq:neccesary-lam}) is uniquely determined by  
\begin{equation}\label{eq:LTVsoln}
\lambda_i(t)=\mathrm{e}^{(t_f-t)\Lambda^\top}\lambda_i(t_f).
\end{equation} 
Substituting this solution in (\ref{eq:neccesary-u}) yields
\begin{equation}\label{eq:Nash00}
   u_i(t)=-\frac{1}{r_i}B_i^\top\mathrm{e}^{(t_f-t)\Lambda^\top}\lambda_i(t_f). 
\end{equation}

Substituting (\ref{eq:neccesary-u}) into the opinion dynamics in (\ref{eq:nonstub-opt}) and then using (\ref{eq:Nash00}), I get
\begin{align}
    \label{eq:dynamics-lambda}
    \dot{x}(t)&=\Lambda x(t)-\sum_{i=1}^nS_i\lambda_i(t) \nonumber\\
    &=\Lambda x(t)-\sum_{i=1}^n S_i\mathrm{e}^{(t_f-t)\Lambda^\top}\lambda_i(t_f).
\end{align}
The solution of (\ref{eq:dynamics-lambda}) at $t$ is given by
\begin{equation}\label{eq:dynamics-sol-00}
    x(t)=\mathrm{e}^{t\Lambda}x_0-\sum_{i=1}^n\Psi_i(t)\lambda_i(t_f), 
\end{equation}
where $\Psi_i(t)$ is defined in (\ref{eq:matrix-Psi-i}).
Using notation (\ref{eq:matrix-Psi}), equation (\ref{eq:dynamics-sol-00}) is rewritten as
\begin{equation}\label{eq:dynamics-sol}
    x(t)=\mathrm{e}^{t\Lambda}x_0-\Psi(t)\lambda(t_f). 
\end{equation}

Stacking (\ref{eq:Lambda-compact}) for $i = 1, \cdots, n$ yields
\begin{equation}\label{eq:lamtf-compact}
    \lambda(t_f)=\Delta x(t_f)
\end{equation}
with $\Delta$ defined in (\ref{eq:matrix-Delta}) and $\lambda(t_f)=[\lambda_1^\top(t_f),\cdots,\lambda_n^\top(t_f)]^\top$. Substituting $\lambda(t_f)$ from (\ref{eq:lamtf-compact}) then into (\ref{eq:dynamics-sol}) at $t_f$ yields 
 \begin{align*}
x(t_f)=\mathrm{e}^{t_f\Lambda}x_0-\Psi(t_f)\Delta x(t_f),
\end{align*}
which can be rewritten as
 \begin{align}\label{eq:state-equation-simplify}
\big(I+\Psi(t_f)\Delta\big)x(t_f)=\mathrm{e}^{t_f\Lambda}x_0.
\end{align}

Using the notation $H(t_f)$ in (\ref{eq:matrix-H}), equation (\ref{eq:state-equation-simplify}) is rewritten as
 \begin{align}\label{eq:state-x(tf)}
x(t_f)=H^{-1}(t_f)\mathrm{e}^{t_f\Lambda}x_0.
\end{align}

If the game has a unique open-loop Nash equilibrium, then (\ref{eq:state-equation-simplify}) is satisfied for any arbitrary $x_0$ and $x(t_f)$. Equivalently, if matrix $H^{-1}(t_f)$ has an inverse for any arbitrary $x(t_f)$, the unique equilibrium actions exist and could be calculated for all $t\in[0,t_f]$. By substituting (\ref{eq:state-x(tf)}) into (\ref{eq:Lambda-compact}) and re-substituting (\ref{eq:Lambda-compact}) in (\ref{eq:neccesary-u}), I obtain (\ref{eq:Nash-nonstub}). 

Substituting $\lambda(t_f)$ from (\ref{eq:lamtf-compact}) and then re-submitting $x(t_f)$ from (\ref{eq:state-x(tf)}), I have 
\begin{align}\label{eq:final-tr}
x(t)&=\mathrm{e}^{t\Lambda}x_0 -\Psi(t)\lambda(t_f) \nonumber \\
&=\mathrm{e}^{t\Lambda}x_0-\Psi(t)\Delta x(t_f)\\ \nonumber 
&=\mathrm{e}^{t\Lambda}x_0-\Psi(t)\Delta H^{-1}(t_f)\mathrm{e}^{t_f\Lambda}x_0, \nonumber 
\end{align}
or in its final form (\ref{eq:Nash-tr-nonstub}). 
This concludes the proof.
\end{proof}

\begin{remark}\label{remark:dist} From the definition of the Laplacian matrix $L_i$ in (\ref{eq:Laplacian}), it can be easily figured out that all nonzero elements of $L_i$ are only in the $i$th row and column. Moreover, the nonzero elements of the $i$th row and column are at the indices $j\in\mathcal{N}_i$.
Based on this structure, it can be deduced that the matrix product $L_iH^{-1}(t_f)\mathrm{e}^{t_f\Lambda}x_0$ only requires the initial $x_{j0}$ entries for $\forall j\in\mathcal{N}_i$ in $x_0$. Thus, I can draw the conclusion that the equilibrium actions (\ref{eq:Nash-nonstub}) are distributed in the sense that each agent only uses local information of her own and her graph neighbors without using any global information of the communication graph.
\end{remark}

\subsection{Global Optimal}
In the game-based optimization~(\ref{eq:nonstub-opt}), each individual by minimizing her cost function attains the locally optimal Nash equilibrium. However, the game equilibrium in general does not correspond to the global social optimum, which minimizes the sum of all costs. 

Let $R=\mathrm{diag}(r_1,\cdots,r_n)$, $B=\mathrm{diag}(b_1,\cdots,b_n)$, and $u(t)=[u_1(t),\cdots,u_n(t)]^\top$. The global optimization problem for the social network is 
\begin{align}\label{eq:nonstub-opt-g}
    &\min_{u}~  J(u(t))= x^\top(t_f)Lx(t_f)+\int_{0}^{t_f}u^\top(t)Ru(t)~\mathrm{dt},\\
    &{s.t.}\notag\\
     &\dot{x}(t)=\Lambda x(t)+ Bu(t), \quad x(0)=x_0. \nonumber
\end{align}

Define the Hamiltonian 
\begin{equation}\label{eq:Hamiltonian2}
        \mathcal{H}=u^\top(t)Ru(t)+\lambda^\top(t)\Big(\Lambda x(t)+Bu(t)\Big). 
\end{equation}
Applying the necessary conditions on the Hamiltonian yield
\begin{align}
    &u(t)=-R^{-1}B^\top\lambda(t), \label{eq:neccesary-u-op}\\
    &\dot{\lambda}(t)=-\Lambda^\top\lambda(t),\quad  \lambda(t_f)=Lx(t_f). \label{eq:neccesary-lam-op}
\end{align}

The solution of (\ref{eq:neccesary-lam-op}) is  
\begin{equation}\label{eq:LTVsoln-op}
\lambda(t)=\mathrm{e}^{(t_f-t)\Lambda^\top}\lambda(t_f).
\end{equation} 
Substituting this solution in (\ref{eq:neccesary-u-op}) yields the optimal control actions
\begin{equation}\label{eq:optimal00}
   u(t)=-R^{-1}B^\top\mathrm{e}^{(t_f-t)\Lambda^\top}\lambda(t_f). 
\end{equation}

Substituting (\ref{eq:neccesary-u-op}) into the opinion dynamics in (\ref{eq:nonstub-opt-g}) and then using (\ref{eq:neccesary-u-op}), I get
\begin{align}
    \label{eq:dynamics-lambda-op}
    \dot{x}(t)&=\Lambda x(t)-BR^{-1}B^\top\lambda(t) \nonumber\\
    &=\Lambda x(t)-BR^{-1}B^\top\mathrm{e}^{(t_f-t)\Lambda^\top}\lambda(t_f).
\end{align}
The solution of (\ref{eq:dynamics-lambda-op}) at $t$ is given by
\begin{equation}\label{eq:dynamics-sol0}
    x(t)=\mathrm{e}^{t\Lambda}x_0-\hat{\Psi}(t)\lambda(t_f), 
\end{equation}
where
\begin{equation*}
    \hat{\Psi}(t)=\int_0^{t}\mathrm{e}^{(t-\tau)\Lambda}BR^{-1}B^\top\mathrm{e}^{(t-\tau)\Lambda^\top}~\mathrm{d\tau}.
\end{equation*}

Substituting $\lambda(t_f)$ from (\ref{eq:neccesary-lam-op}) then into (\ref{eq:dynamics-sol0}) at $t_f$ yields 
 \begin{align*}
x(t_f)=\mathrm{e}^{t_f\Lambda}x_0-\hat{\Psi}(t_f)L x(t_f),
\end{align*}
which can be rewritten as
 \begin{align*}
\big(I+\hat{\Psi}(t_f)L\big)x(t_f)=\mathrm{e}^{t_f\Lambda}x_0,
\end{align*}
or equivalently,
 \begin{align}\label{eq:state-x(tf)-op}
x(t_f)=\hat{H}^{-1}(t_f)\mathrm{e}^{t_f\Lambda}x_0,
\end{align}
where
\begin{equation*}
    \hat{H}(t_f)=I+\hat{\Psi}(t_f)L.
\end{equation*}

Finally, the associated opinion trajectory with optimal control actions for the global opinion formation problem in (\ref{eq:nonstub-opt-g}) is given by 
\begin{align}
            x(t)&=\big(\mathrm{e}^{t\Lambda}-\hat{\Psi}(t)L \hat{H}^{-1}(t_f)\mathrm{e}^{t_f\Lambda}\big)x_0.  \label{eq:optimal-tr}
\end{align}
Substituting $\lambda(t_f)$ from (\ref{eq:neccesary-lam-op}) into (\ref{eq:optimal00}) and then using (\ref{eq:state-x(tf)-op}), the control actions are
 \begin{align}
         u(t&)=-R^{-1}B^\top\mathrm{e}^{(t_f-t)\Lambda^\top}L\hat{H}^{-1}(t_f)\mathrm{e}^{t_f\Lambda}x_0. \label{eq:optimal}
 \end{align}

The nonsingularity of $\hat{H}(t_f)$ is a prerequisite for (\ref{eq:optimal-tr}) and (\ref{eq:optimal}). 

\section{Simulation Results}\label{simulation}

In this section, I apply the theoretical results to the well-known social network of Zachary’s Karate Club~\cite{Zachary}. With its underlying social graph given in Fig.~\ref{fig:karate}, Zachary’s Karate Club network, created by Wayne Zachary in 1977, is a social network between the members of a karate club at a US university. This network has 34 nodes and 78 edges, where the nodes and edges represent the club members and their mutual friendships. I show the consequences of the execution of game strategies in Theorem~\ref{theorem:nonstub} on the evolution of opinions for Zachary’s Karate Club network.  
\begin{figure}
\centering
       \includegraphics[width=0.5\textwidth]{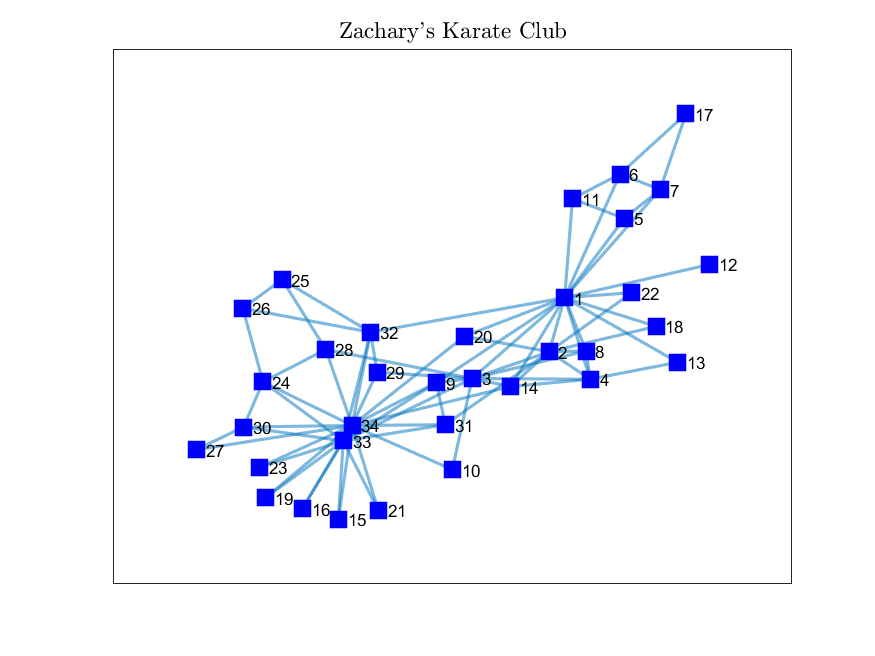}
\caption{Zachary's Karate Club network. }
\label{fig:karate}
\end{figure} 

The simulations are carried out for the time interval of $t_f=10$ and the initial opinion vector $x_0$ for each experiment is selected from a uniform distribution in $(-1.5,-0.5)\cup(0.5,1.5)$. Such a distribution of the initial opinions shows that the population's initial opinions are split into two groups or clusters. The evolution of club members' opinions under the continuous time HK model (\ref{eq:dynamics0}) prior to any optimization (i.e., $u_i(t)=0$) is according to $x(t)=\mathrm{e}^{t\Lambda}x_0$. The corresponding opinion trajectories are illustrated in Fig.~\ref{fig:HK}. As it is seen, the club members' opinions reach nearly a consensus about the average opinion in the network. I should note that a consensus is realized at a long enough horizon length.
\begin{figure}
\centering
       \includegraphics[width=0.32\textwidth]{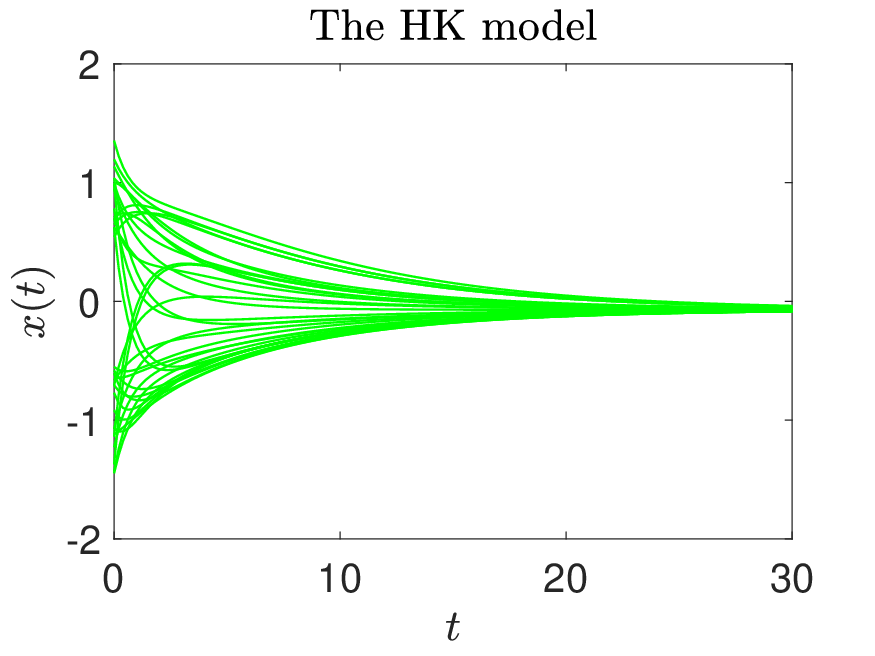}
\caption{Evolution of opinion trajectories in the HK model for Zachary’s network ($u_i(t)=0~\forall i\in\mathcal{V}$).}
\label{fig:HK}
\end{figure} 

For optimizing the opinions, in the continuous-time HK model (\ref{eq:dynamics0}), I let $b_i=1$. For Zachary's Karate Club social network, the opinion trajectories as a result of optimization (\ref{eq:nonstub-opt}) and (\ref{eq:nonstub-opt-g}) are shown in Fig.~\ref{fig:nonstub}, by the green and red colors, respectively. According to this figure, the non-stubborn network under the framework of global optimization reaches a consensus on the average opinions for $r_i=1$. Therefore, a consensus is the global social optimum norm. Using their game strategies, the network members' final opinions have moved toward each other, but they have not reached a consensus. By imposing large $r_i$, the network members seek to minimize their influence effort rather than their disagreement with others. This is because, given a sufficiently large $r_i$, the related optimization is reduced to the minimization of the influence effort term. As it is seen from Fig.~\ref{fig:nonstub}, for a relatively large $r_i=20~\forall i\in\mathbf{V}$, the final opinion trajectories have come closer to each other, but not as much as in the case for $r_i=1$. 
\begin{figure}
\centering
       \includegraphics[width=0.32\textwidth]{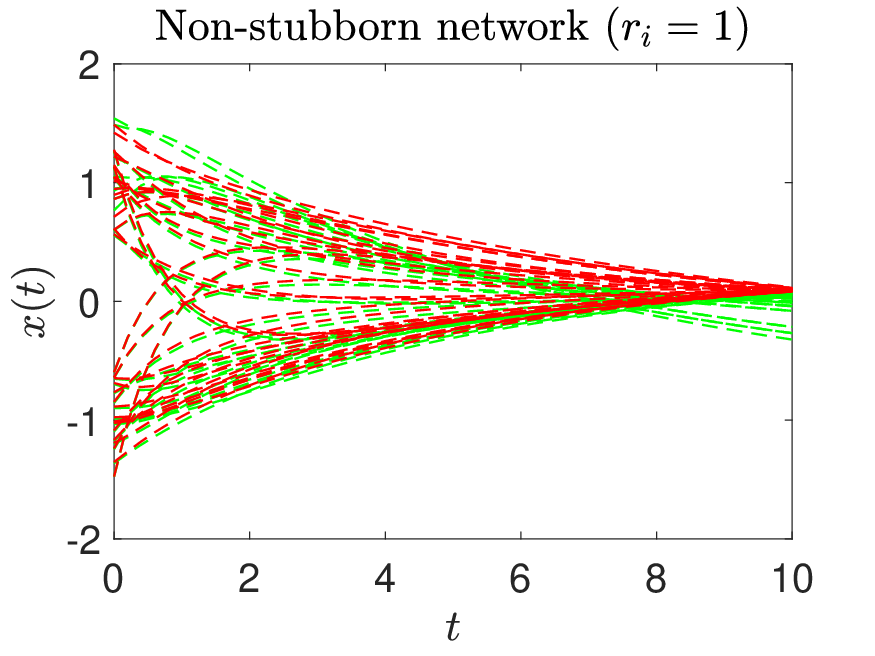}
       \includegraphics[width=0.32\textwidth]{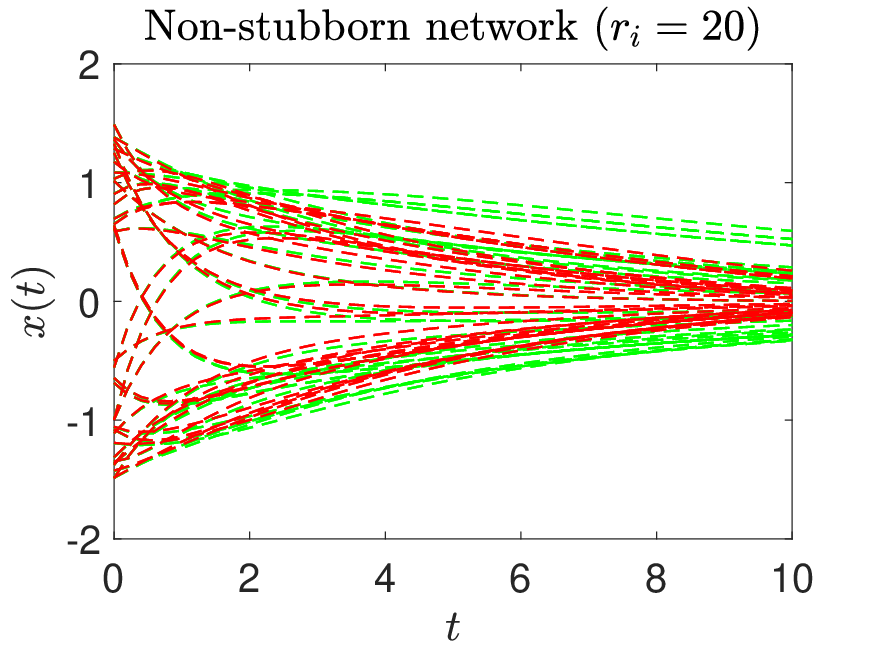}
\caption{Opinion trajectories associated with the game strategies and global optimal actions for the non-stubborn Zachary’s network, represented by the green and red colors, respectively.}
\label{fig:nonstub}
\end{figure} 

The finite horizon length is an important factor in the formation of final opinions. From Fig.~\ref{fig:karate} and Fig.~\ref{fig:nonstub}, it is seen that the optimization reduces the finite horizon length for final opinion formation significantly. Fig.~\ref{fig:horizon} shows that a consensus emerges from the optimization of opinions at a long horizon length.  
\begin{figure}
\centering
       \includegraphics[width=0.32\textwidth]{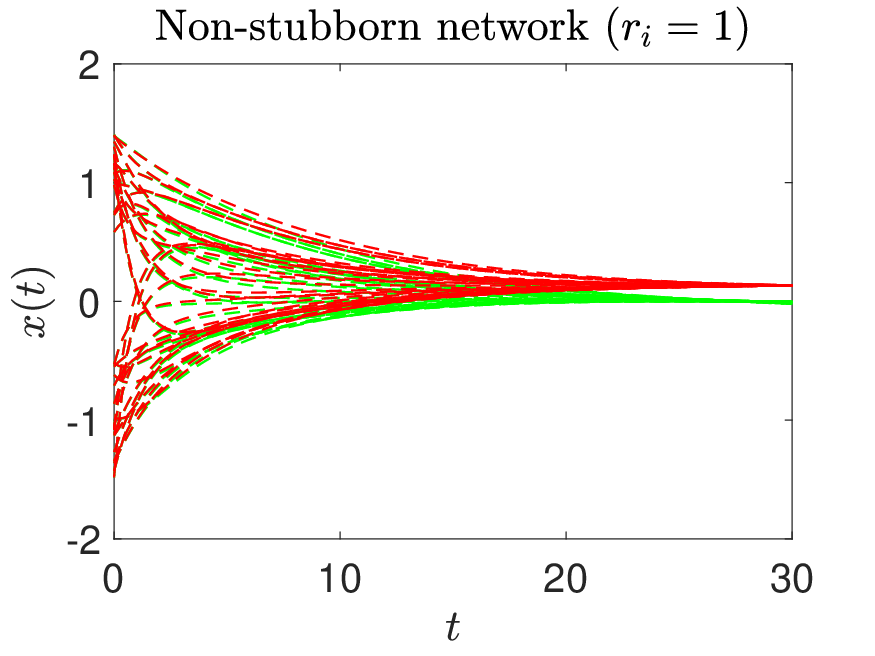}
\caption{A consensus of final opinions is on the horizon.}
\label{fig:horizon}
\end{figure} 

\section{Conclusions}\label{conclusions}
This paper studied optimizing opinions of the Hegselmann-Krause model in a game-theoretic framework. The execution of game equilibrium strategies in the well-known Zachary’s Karate Club social network with non-stubborn individuals showed that the opinions came close to each other but that a consensus of final opinions did not emerge. However, the club members might be stubborn individuals whose prejudices influence the formation of their final opinions. A future research direction could be optimizing opinions, considering the stubbornness of individuals in the game-theoretic analysis of social networks. Feedback Nash equilibrium strategies for time-varying communication graphs and confidence bound values can be investigated.

\section*{Acknowledgment}

This work was supported by SGS, V\v{S}B - Technical University of Ostrava, Czech Republic, under grant No. SP2023/12 “Parallel processing of Big Data X”.

\end{document}